\documentclass[11pt,letterpaper]{article}

\usepackage{amsmath,amsfonts,amsthm,amssymb,stmaryrd,graphicx,relsize,bm}  
\theoremstyle{plain}

\numberwithin{equation}{section}

\newtheorem{thm}{Theorem}[section]
\newtheorem{lem}[thm]{Lemma}
\newtheorem{cor}[thm]{Corollary}

\newenvironment{exam}[1]
{\begin{flushleft}\textbf{Example #1}.\enspace}%
{\end{flushleft}}

\allowdisplaybreaks  

\newcommand{\tbullet}{\mathrel{\raise .4ex\hbox{\tiny$\bullet$}}} 

\newcommand{\rmtr}{\mathrm{tr\,}}
\newcommand{\rmin}{\mathrm{In\,}}

\newcommand{\ascript}{\mathcal{A}}
\newcommand{\cscript}{\mathcal{C}}
\newcommand{\escript}{\mathcal{E}}
\newcommand{\iscript}{\mathcal{I}}
\newcommand{\jscript}{\mathcal{J}}

\newcommand{\lscript}{\mathcal{L}}
\newcommand{\mscript}{\mathcal{M}}
\newcommand{\oscript}{\mathcal{O}}
\newcommand{\pscript}{\mathcal{P}}
\newcommand{\sscript}{\mathcal{S}}

\newcommand{\alphahat}{\widehat{\alpha}}
\newcommand{\betahat}{\widehat{\beta}}
\newcommand{\iscripthat}{\widehat{\iscript}}
\newcommand{\jscripthat}{\widehat{\jscript}}
\newcommand{\mscripthat}{\widehat{\mscript}}
\newcommand{\alphatilde}{\widetilde{a}}
\newcommand{\capctilde}{\widetilde{C}}

\newcommand{\ab}[1]{\left|#1\right|}
\newcommand{\doubleab}[1]{\left|\left|#1\right|\right|}
\newcommand{\brac}[1]{\left\{#1\right\}}
\newcommand{\paren}[1]{\left(#1\right)}
\newcommand{\sqbrac}[1]{\left[#1\right]}
\newcommand{\elbows}[1]{{\left\langle#1\right\rangle}}
\newcommand{\ket}[1]{{\left|#1\right>}}
\newcommand{\bra}[1]{{\left<#1\right|}}

\begin{document}

\title{PARTS AND COMPOSITES\\OF\\QUANTUM SYSTEMS}
\author{Stan Gudder\\ Department of Mathematics\\
University of Denver\\ Denver, Colorado 80208\\
sgudder@du.edu}
\date{}
\maketitle

\begin{abstract}
We consider three types of entities for quantum measurements. In order of generality, these types are: observables, instruments and measurement models. If $\alpha$ and $\beta$ are entities, we define what it means for $\alpha$ to be a part of $\beta$. This relationship is essentially equivalent to $\alpha$ being a function of $\beta$ and in this case $\beta$ can be employed to measure $\alpha$. We then use the concept to define coexistence of entities and study its properties. A crucial role is played by a map $\alphahat$ which takes an entity of a certain type to one of lower type. For example, if $\iscript$ is an instrument, then $\iscripthat$ is the unique observable measured by $\iscript$. Composite systems are discussed next. These are constructed by taking the tensor product of the Hilbert spaces of the systems being combined. Composites of the three types of measurements and their parts are studied. Reductions of types to their local components are discussed. We also consider sequential products of measurements. Specific examples of L\"uders, Kraus and trivial instruments are used to illustrate various concepts. We only consider finite-dimensional systems in this article.
\end{abstract}

\section{Introduction}  
Two important operations on quantum systems are the formations of parts and composites. In a rough sense, these operations are opposites to each other. The parts of a measurement $\alpha$ are smaller components of $\alpha$ in the sense that they can be simultaneously measured by $\alpha$. A composite system is a combination of two or more other systems. This combination is formed using the tensor product $H=H_1\otimes H_2$ where $H_1$ and $H_2$ are the Hilbert spaces describing two subsystems. The composite system contains more information than the individual systems because $H$ describes how $H_1$ and $H_2$ interact. We can reduce measurements on $H$ to simpler ones on $H_1$ and $H_2$ but information is lost in the process.

Section~2 presents the basic definitions that are needed in the sequel. Three types of quantum measurements are considered. In order of generality, these types are: observables, instruments and measurement models. At the basic level is an observable $A$ which is a measurement whose outcome probabilities $\rmtr (\rho A_x)$ are determined by the state $\rho$ of the system. At the next level is an instrument $\iscript$. We think of $\iscript$ as an apparatus that can be employed to measure an observable $\iscripthat$. Although $\iscripthat$ is unique, there are many instruments that can be used to measure an observable. Moreover, $\iscript$ gives more information than $\iscripthat$ because, depending on the outcome $x$, $\iscript$ updates the input state $\rho$ to the output state $\iscript _x(\rho )/\rmtr (\rho\iscripthat _x)$. At the highest level is a measurement model $\mscript$ that measures a unique instrument $\mscripthat$. Again, there are many measurement models that measure an instrument and $\mscript$ contains more detailed information. For conciseness, we call these types of instruments \textit{entities}. We should mention that all the quantum systems in this article are assumed to be finite-dimensional.

Section~3 considers system parts. If $\alpha$ and $\beta$ are entities, we define what it means for $\alpha$ to be a \textit{part} of $\beta$ and when this is the case, we write $\alpha\to\beta$. If $\alpha\to\beta$ and $\beta\to\alpha$, we say that $\alpha$ and $\beta$ are
\textit{equivalent}. We show that $\alpha\to\beta$ implies $\alphahat\to\betahat$ and that $\to$ is a partial order to within equivalence. The relation $\alpha\to\beta$ is the same as $\alpha$ being a function of $\beta$ or $\betahat$ and in this case, $\beta$ can be employed to measure $\alpha$. We then use this concept to define coexistence of entities and study its properties. We show that joint measurability is equivalent to coexistence. We then introduce sequential products of observables and use this concept to illustrate parts of entities.

Section~4 discusses composite systems. These are constructed by taking the tensor product $H=H_1\otimes H_2$ where $H_1,H_2$ are the Hilbert spaces of the systems being combined. Composites of the three types of measurements and parts of these composites are studied. Reductions of types into their local components are discussed. Specific examples of L\"uders, Kraus and trivial instruments are employed to illustrate various concepts.

\section{Basic Definitions}  
This section discusses the basic concepts and definitions that are needed in the sequel. Since these ideas are well developed in the literature \cite{bgl95,fhl18,hz12,kra83,nc00}, we shall proceed quickly and leave details and motivation to the reader's discretion. In this article we shall only consider finite-dimensional complex Hilbert spaces $H$. Let $\lscript (H)$ be the set of linear operators on $H$. For $S,T\in\lscript (H)$ we write $S\le T$ if $\elbows{\phi ,S\phi}\le\elbows{\phi ,T\phi}$ for all $\phi\in H$. We define the set of \textit{effects} by
\begin{equation*} 
\escript (H)=\brac{a\in\lscript (H)\colon 0\le a\le 1}
\end{equation*}
where $0,1$ are the zero and identity operators, respectively. Effects correspond to yes-no measurements and when the result of measuring
$a$ is yes, we say that $a$ \textit{occurs}. The complement of $a\in\escript (H)$ is $a'=1-a$ and $a'$ occurs if and only if $a$ does not occur. A one-dimensional projection $P_\phi =\ket{\phi}\bra{\phi}$, where $\doubleab{\phi}=1$ is an effect called an \textit{atom}. We call
$\rho\in\escript (H)$ a \textit{partial state} if $\rmtr(\rho )\le 1$ and $\rho$ is a \textit{state} if $\rmtr (\rho )=1$. We denote the set of partial states by $\sscript _p(H)$ and the set of states by $\sscript (H)$. If $\rho\in\sscript (H)$, $a\in\escript (H)$, we call $\pscript _\rho (a)=\rmtr (\rho a)$ the \textit{probability that} $a$ \textit{occurs} in the state $\rho$ \cite{bgl95, hz12,nc00}. For $a,b\in\escript (H)$, their \textit{sequential product} is the effect $a\circ b=a^{1/2}ba^{1/2}$ where $a^{1/2}$ is the unique square root of $a$ \cite{gg02,gn01,gud120}. We interpret $a\circ b$ as the effect that results from first measuring $a$ and then measuring $b$. We also call $a\circ b$ the effect $b$ \textit{conditioned on} the effect $a$ and write $(b\mid a)=a\circ b$.

Let $\Omega _A$ be a finite set. A (finite) \textit{observable with outcome-space} $\Omega _A$ is a subset
\begin{equation*} 
A=\brac{A_x\colon x\in\Omega _A}\subseteq\escript (H)
\end{equation*}
satisfying $\sum\limits _{x\in\Omega _A}A_x=1$. We denote the set of observables on $H$ by $\oscript (H)$. If
$B=\brac{B_y\colon y\in\Omega _B}$ is another observable, we define the \textit{sequential product} $A\circ B\in\oscript (H)$ \cite{gud120,gud220,gud320} to be the observable with outcome-space $\Omega _A\times\Omega _B$ given by
\begin{equation*} 
A\circ B=\brac{A_x\circ B_y\colon (x,y)\in\Omega _A\times\Omega _B}
\end{equation*}
We also define the observable $B$ \textit{conditioned by} $A$ as
\begin{equation*} 
(B\mid A)=\brac{(B\mid A)_y\colon y\in\Omega _B}\subseteq\escript (H)
\end{equation*}
where $(B\mid A)_y=\sum\limits _{x\in\Omega _A}(A_x\circ B_y)$. If $A\in\oscript (H)$ we define the \textit{effect-valued measure} (or POVM) $X\to A_X$ from $2^{\Omega _A}$ to $\escript (H)$ by $A_X=\sum\limits _{x\in X}A_x$ and we also call $X\mapsto A_X$ an \textit{observable} \cite{gud120,hz12,nc00}. Moreover, we have the observables
\begin{align*}
(A\circ B)_\Delta =\sum _{(x,y)\in\Delta}(A_x\circ B_y)\\
\intertext{and}
(B\mid A)_Y=\sum _{x\in\Omega _A}(A_x\circ B_Y)
\end{align*}
If $\rho\in\sscript (H)$ and $A\in\oscript (H)$, the \textit{probability that} $A$ \textit{has an outcome in} $X\subseteq\Omega _A$ when the system is in state $\rho$ is $\pscript _\rho (A_X)=\rmtr (\rho A_X)$. Notice that $X\mapsto\pscript _\rho (A_X)$ is a probability measure on
$\Omega _A$. We call
\begin{equation*} 
\pscript _\rho (A_X\hbox{ then }B_Y)=\rmtr\sqbrac{\rho (A\circ B)_{X\times Y}}
\end{equation*}
the \textit{joint probability} of $A_X$ then $B_Y$ \cite{gud120,gud220,gud320}.

An \textit{operation} is a completely positive map $\ascript\colon\sscript _p(H)\to\sscript _p)(H)$ \cite{bgl95,hz12,nc00}. Any operation has a \textit{Kraus decomposition}
\begin{equation*} 
\ascript (\rho )=\sum _{i=1}^nS_i\rho S_i^*
\end{equation*}
where $S_i\in\lscript (H)$ with $\sum\limits _{i=1}^nS_i^*S_i\le 1$. An operation $\ascript$ is a \textit{channel} if $\ascript (\rho )\in\sscript (H)$ for all $\rho\in\sscript (H)$. In this case $\sum\limits _{i=1}^nS_i^*S_i=1$ and we denote the set of channels on $H$ by $\cscript (H)$. Notice
that if $a\in\escript (H)$, then $\rho\mapsto (\rho\mid a)=a\circ\rho$ is an operation and if $A\in\oscript (H)$, then
$\rho\mapsto (\rho\mid A)=\sum\limits _{x\in\Omega _A}(A_x\circ\rho )$ is a channel. For a finite set $\Omega _\iscript$, a (finite)
\textit{instrument} with outcome-space $\Omega _\iscript$ is a set of operations $\iscript =\brac{\iscript _x\colon x\in\Omega _\iscript}$ satisfying
$\cscript _\iscript =\sum\limits _{x\in\Omega _\iscript}\iscript _x\in\cscript (H)$ \cite{bgl95,hz12,nc00,oz01}. Defining
$\iscript _X=\sum\limits _{x\in X}\iscript _x$ for $X\subseteq\Omega _\iscript$, we see that $X\mapsto\iscript _X$ is an operation-valued measure on $H$ that we also call an \textit{instrument}. We denote the set of instruments on $H$ by $\rmin (H)$. We say that
$\iscript\in\rmin (H)$ \textit{measures} $A\in\oscript (H)$ if $\Omega _A=\Omega _\iscript$ and
\begin{equation}                
\label{eq21}
\pscript _\rho (A_X)=\rmtr\sqbrac{\iscript _X(\rho )}
\end{equation}
for every $\rho\in\sscript (H)$, $X\subseteq\Omega _A$. There is a unique $A\in\oscript (H)$ that $\iscript$ measures and we write
$A=\iscripthat$ \cite{bgl95,hz12,oz01}. For $\iscript ,\jscript\in\rmin (H)$, we define the \textit{product instrument} with outcome space
$\Omega _\iscript\times\Omega _\jscript$ by
\begin{equation*} 
(\iscript\circ\jscript )_{(x,y)}(\rho )=\jscript _y\sqbrac{\iscript _x(\rho )}
\end{equation*}
for every $\rho\in\sscript (H)$. We also define the \textit{conditioned instrument} with outcome-space $\Omega _\jscript$ by
\begin{equation*} 
(\jscript\mid\iscript )_y=\sum _{x\in\Omega _\iscript}(\iscript\circ\jscript )_{(x,y)}=\jscript _y\sqbrac{\cscript _\iscript (\rho )}
\end{equation*}
We conclude that
\begin{equation*} 
(\iscript\circ\jscript )_\Delta (\rho )=\sum _{(x,y)\in\Delta}\jscript _y\paren{\iscript _x(\rho )}
\end{equation*}
for all $\Delta\subseteq\Omega _\iscript\times\Omega _\jscript$ and
\begin{equation*} 
(\jscript\mid\iscript )_Y=\sum _{y\in Y}\jscript _y\paren{\cscript _\iscript (\rho )}
\end{equation*}
for all $Y\subseteq\Omega _\jscript$ \cite{gud120,gud220,gud320}.

A \textit{finite measurement model} (MM) is a 5-tuple $\mscript =(H,K,\eta ,\nu ,F)$ where $H$, $K$ are finite-dimensional Hilbert spaces called the \textit{base} and \textit{probe systems}, respectively, $\eta\in\sscript (K)$ is an \textit{initial probe state}, $\nu\in\cscript (H\otimes K)$ is a channel describing the measurement interaction between the base and probe systems and $F\in\oscript (K)$ is the \textit{probe} (or \textit{meter})
\textit{observable} \cite{bgl95,hz12,hrsz09}, We say that $\mscript$ \textit{measures the model instrument} $\mscripthat\in\rmin (H)$ where
$\mscripthat$ is the unique instrument satisfying
\begin{equation}                
\label{eq22}
\mscripthat _X(\rho )=\rmtr _K\sqbrac{\nu (\rho\otimes\eta )(I\otimes F_X)}
\end{equation}
for all $\rho\in\sscript (H)$, $X\subseteq\Omega _F$. In \eqref{eq22}, $\rmtr _K$ is the partial trace over $K$ \cite{hz12,nc00}. We also say that $\mscript$ \textit{measures the model observable} $\mscript ^{\wedge\wedge}$.

We thus have three levels of abstraction. At the basic level is an observable $A$ which is a measurement whose outcome probabilities
$\rmtr (\rho A_x)$ are determined by the state $\rho$ of the system. At the next level is an instrument $\iscript$. We think of $\iscript$ as an apparatus that can be employed to measure an observable $\iscripthat$. Although $\iscripthat$ is unique, there are many instruments that can be used to measure an observable. Moreover, $\iscript$ gives more information than $\iscripthat$ because, depending on the outcome $x$
(or event $X$), $\iscript$ updates the input state $\rho$ to the output partial state $\iscript _x(\rho )$ (or $\iscript _X(\rho )$). At the highest level is a measurement model $\mscript$ that measures a unique model instrument $\mscripthat$ and a unique model observable
$\mscript ^{\wedge\wedge}$. Again, there are many $MM$s that measure any instrument or observable and $\mscript$ contains more detailed information on how the measurement is performed.

\section{System Parts}  
We begin by discussing parts of systems at the three levels considered in Section~2. We then show how parts can be used to define coexistence at these levels and even between levels. We also show that coexistence is equivalent to simultaneous measurability.

An element at one of the three levels discussed in Section~2 is called an \textit{entity}. The three levels are said to be the \textit{types} 1, 2 and 3, respectively. The concept of an entity being part of another entity was originally introduced in \cite{hrsz09,hmr14}. If $A,B\in\oscript (H)$, we say that $A$ is \textit{part of} $B$ (and write $A\to B$) if there exists a surjection $f\colon\Omega _B\to\Omega _A$ such that
$A_x=B_{f^{-1}(x)}$ for all $x\in\Omega _A$. We then write $A=f(B)$. It follows that $A_X=B_{f^{-1}(X)}$ for all $X\in\Omega _A$ and that
\begin{equation}                
\label{eq31}
A_X=\sum\brac{B_y\colon f(y)\in X}
\end{equation}
If $\iscript ,\jscript\in\rmin (H)$, we say that $\iscript$ is \textit{part of} $\jscript$ (and write $\iscript\to\jscript$) if there exists a surjection
$f\colon\Omega _\jscript\to\Omega _\iscript$ such that $\iscript _x=\jscript _{f^{-1}(x)}$ for all $x\in\Omega _\jscript$. We then write
$\iscript =f(\jscript )$ and an equation analogous to \eqref{eq31} holds. For $MM$s $\mscript _1=(H,K,\eta ,\nu ,F_1)$ and
$\mscript _2=(H,K,\eta ,\nu ,F_2)$ we say that $\mscript _1$ is \textit{part of} $\mscript _2$ (and write $\mscript _1\to\mscript _2$) if
$F_1\to F_2$. It follows that $F_1=f(F_2)$ and we write $\mscript _1=f(\mscript _2)$. We can also define ``part of'' for entities of different types. An observable $A\in\oscript (H)$ is \textit{part of} $\iscript\in\rmin (H)$ (written $A\to\iscript$) if $A\to\iscripthat$ and $A$ is \textit{part of}
$\mscript$ (written $A\to\mscript$) if $A\to\mscripthat$ which is equivalent to $A\to\mscript ^{\wedge\wedge}$. Finally, we say that $\iscript$ is \textit{part of} $\mscript$ (written $\iscript\to\mscript$) if $\iscript\to\mscripthat$. Two entities $\alpha$ and $\beta$ are \textit{equivalent}
(written $\alpha\cong\beta$) if $\alpha\to\beta$ and $\beta\to\alpha$. It is easy to check that $\cong$ is an equivalence relation and that
$\alpha\cong\beta$ if and only if $\alpha =f(\beta )$ for $f$ a bijection. Our first result summarizes properties possessed by ``part of''. Some of these properties have been verified in \cite{hmr14}, but we give the full proof for completeness.

\begin{thm}    
\label{thm31}
{\rm{(a)}}\enspace If $\alpha ,\beta$ are of types 2 or 3 and $\alpha\to\beta$, then $\alphahat\to\betahat$.
{\rm{(b)}}\enspace $f(\iscripthat )=f(\iscript )^\wedge$ and $f(\mscripthat )=f(\mscript )^\wedge$.
{\rm{(c)}}\enspace If $\alpha ,\beta ,\gamma$ are of the same type and $\alpha =g(\beta )$, $\beta =f(\gamma )$, then
$\alpha =(g\circ f)(\gamma )$.
{\rm{(d)}}\enspace The relation $\to$ is a partial order to within equivalence.
{\rm{(e)}}\enspace If $\alpha$ and $\beta$ are of different types and $\alpha\to\beta$, then $\alpha =\betahat _1$ where $\beta _1\to\beta$.
\end{thm}
\begin{proof}
(a)\enspace Let $\iscript ,\jscript\in\rmin (H)$ with $\iscript\to\jscript$. Then there exists a surjection $f\colon\Omega _\jscript\to\Omega _\iscript$ such that $\iscript = f(\jscript )$. We now show that $\iscripthat =f(\jscripthat\,)$. Indeed, for any $\rho\in\sscript (H)$, $x\in\Omega _\iscript$ we have that
\begin{equation*}
\rmtr (\rho\iscripthat _x)=\rmtr\sqbrac{\iscript _x(\rho )}=\rmtr\sqbrac{\jscript _{f^{-1}(x)}(\rho )}=\rmtr\sqbrac{\rho\jscript _{f^{-1}(x)}}
   =\rmtr\sqbrac{\rho f(\jscript )_x}
\end{equation*}
Hence, $\iscript =f(\jscripthat\,)$ so $\iscripthat\to\jscripthat$. Let $\mscript _1=(H,K,\eta ,\nu ,F_1)$, $\mscript _2(H,K,\eta ,\nu ,F_2)$ be $MM$s where $F_1=f(F_2)$. Then for any $\rho\in\sscript (H)$, $x\in\Omega _{F_1}$ we have that
\begin{align*}
\mscripthat _{1,x}(\rho )&=\rmtr _K\sqbrac{\nu (\rho\otimes\eta )(I\otimes F_{1,x})}
  =\rmtr _K\sqbrac{\nu (\rho\otimes\eta )(I\otimes F_{2,f^{-1}(x)})}\\
   &=\mscripthat _{2,f^{-1}(x)}(\rho )=f(\mscripthat _2)
\end{align*}
Hence, $\mscripthat _1=f(\mscripthat _2)$ so $\mscripthat _1\to\mscripthat _2$. If $\iscript\to\mscript$, then $\iscript\to\mscripthat$. As before,
$\iscripthat\to\mscript ^{\wedge\wedge}$ so $\iscripthat\to\mscripthat$.\newline
(b)\enspace This was proved in (a).\enspace (c)\enspace We prove the result for observables $A,B,C$ and the result for instruments and
$MM$s is similar. We have that $A_x=B_{g^{-1}(x)}$ and $B_y=C_{f^{-1}(y)}$. Since $g\colon\Omega _B\to\Omega _A$ and
$f\colon\Omega _C\to\Omega _B$, we have that $g\circ f\colon\Omega _C\to\Omega _A$. Hence,
\begin{equation*}
A_x=B_{g^{-1}(x)}=C_{f^{-1}(g^{-1}(x))}=C_{(g\circ f)^{-1}(x)}
\end{equation*}
Hence, $A=(g\circ f)(C)$.
(d)\enspace We only need to prove that if $\alpha\to\beta$ and $\beta\to\gamma$, then $\alpha\to\gamma$. If $\alpha ,\beta ,\gamma$ are of the same type, the $\alpha\to\gamma$ follows from (c). Suppose $A,B\in\oscript (H)$, $\iscript\in\rmin (H)$ and $A\to B$, $B\to\iscript$. Then
$A\to B\to\iscripthat$ and these are the same type so $A\to\iscripthat$ and hence, $A\to\iscript$. Suppose $A\in\oscript (H)$
$\iscript ,\jscript\in\rmin (H)$ and $A\to\iscript$, $\iscript\to\jscript$. Then $A\to\iscripthat$ and $\iscript\to\jscript$. By (a) we have
$\iscripthat\to\jscripthat$. Since $A,\iscripthat ,\jscripthat$ have the same type, $A\to\jscripthat$ and hence, $A\to\jscript$. Suppose that
$A\to\iscript$ and $\iscript\to\mscript$. Then $A\to\iscripthat$ and $\iscript\to\mscripthat$. By (a) $\iscripthat\to\mscript ^{\wedge\wedge}$ so
$A\to\iscripthat$ and $\iscripthat\to\mscript ^{\wedge\wedge}$. Since these are the same type, we have that $A\to\mscript ^{\wedge\wedge}$ so $A\to\mscript$. Similar reasoning holds for the cases $\iscript\to\jscript\to\mscript$ and $\iscript\to\mscript _1\to\mscript _2$.\newline
(e)\enspace If $A\in\oscript (H)$, $\iscript\in\rmin (H)$ and $A\to\iscript$, then $A\to\iscripthat$ so $A=f(\iscripthat )$ for some surjection
$f\colon\Omega _{\jscripthat}\to\Omega _A$. By (b) we have that $f(\iscripthat )=f(\iscript )^\wedge$ so letting $\iscript _1=f(\iscript )$ we have that $A=f(\iscript )^\wedge =\iscripthat _1$. Hence, $\iscript _1\to\iscript$. If $A\to\mscript$, then $A\to\mscript ^{\wedge\wedge}$. By (b),
$A=f(\mscript ^{\wedge\wedge})=\sqbrac{f(\mscripthat )}^\wedge$. Letting $\iscript =f(\mscripthat )$ we have that $A=\iscripthat$,
$\iscript\to\mscripthat\to\mscript$. If $\iscript\to\mscript$, then $\iscript\to\mscripthat$. By (b) $\iscript =f(\mscript ^\wedge )=f(\mscript )^\wedge$. Letting $\iscript _1=f(\mscript )$, we have that $\iscript =\iscripthat _1$ and $\iscript _1\to\mscript$.
\end{proof}

For an entity $\alpha$, we denote its set of parts by $\alphatilde =\brac{\beta\colon\beta\to\alpha}$. We say that a set $\ascript$ of entities
\textit{coexist} if $\ascript\subseteq\alphatilde$ for some entity $\alpha$. A coexistent set $\ascript\subseteq\alphatilde$ is thought of as being simultaneously measured by $\alpha$. A related concept is that of joint measurability. We say that observables $A^i\in\oscript (H)$ with outcome sets $\Omega _i$, $i=1,2,\ldots ,n$ are \textit{jointly measurable} with \textit{joint observable} $B\in\oscript (H)$ if
$\Omega _B=\Omega _1\times\cdots\times\Omega _n$ and for all $x_i\in\Omega _i$ we have
\begin{equation}                
\label{eq32}
A_{x_i}^i=\sum\brac{B_{(x_1,\ldots ,x_i,\ldots ,x_n)}\colon x_j\in\Omega _j, j\ne i}
\end{equation}
We interpret $A^i$ as being the $i$th \textit{marginal} of $B$ as in classical probability theory \cite{gud120,gud220,hrsz09}. Similar definitions can be made for joint measurability of instruments and $MM$s.

\begin{thm}    
\label{thm32}
A set of observables $A^i\in\oscript (H)$, $i=1,2,\ldots ,n$ is jointly measurable if and only if the $A^i$ coexist.
\end{thm}
\begin{proof}
If $\brac{A_i\colon i=1,2,\ldots ,n}$ are jointly measurable, there exists a joint observable $B\in\oscript (H)$ satisfying \eqref{eq32}. Defining
$f_i\colon\Omega _B\to\Omega _Ai$ by
\begin{equation*}
f_i(x_1,\ldots ,x_i,\ldots ,x_n)=x_i
\end{equation*}
for $i=1,2,\ldots ,n$, then by \eqref{eq32} we have that $A_{x_i}^i=B_{f_i^{-1}(x_i)}$ for all $x_i\in\Omega _i$. Hence, $A^i=f_i(B)$, $i=1,2,\ldots ,n$, so $\brac{A^i}$ coexist. Conversely, suppose that $\brac{A^i\colon i=1,2,\ldots ,n}$ coexist so there exists an observable
$C\in\oscript (H)$ such that $A^i\in\capctilde$, $i=1,2,\ldots ,n$. We then have surjections $f_i\colon\Omega _C\to\Omega _{A_i}$ such that
$A^i=f_i(C)$, $i=1,2,\ldots ,n$. Define $\Omega _B=\Omega _1\times\cdots\times\Omega _n$, a surjection
$h\colon\Omega _C\to\Omega _B$ by $h(y)=\paren{f_1(y),\ldots ,f_n(y)}$ and let $B=h(C)$. For $i=1,2,\ldots ,n$, we obtain
\begin{align*}
A_{x_i}^i&=C_{f_i^{-1}(x_i)}=\sum\brac{C_y\colon f_i(y)=x_i}\\
   &=\sum\brac{C_y\colon\paren{f_1(y),\ldots ,f_n(y)}=\brac{x_1,\ldots ,x_i,\ldots ,x_n},x_j\in\Omega _j,j\ne i}\\
   &=\sum\brac{C_y\colon h(y)=(x_1,\ldots ,x_i,\ldots ,x_n)\colon x_j\in\Omega _j, j\ne i}\\
   &=\sum\brac{C_{h^{-1}(x_1,\ldots ,x_i,\ldots ,x_n)}\colon x_j\in\Omega _j,j\ne i}\\
   &=\sum\brac{h(C)_{(x_1,\ldots ,x_i,\ldots ,x_n)}\colon x_j\in\Omega _j,j\ne i}\\
   &=\sum\brac{B_{(x_1,\ldots ,x_i,\ldots ,x_n)}\colon x_j\in\Omega _j, j\ne i}
\end{align*}
Thus, \eqref{eq32} holds so $\brac{A^i}$ are jointly measurable.
\end{proof}

Theorem~\ref{thm32} also holds for instruments and $MM$s. An important property of coexistent entities is that they have joint probability distributions $\Phi _\rho$ for all $\rho\in\sscript (H)$. For example, if $A,B\in\oscript (H)$ coexist, then $A=f(C)$, $B=g(C)$ for some
$C\in\oscript (H)$. Then for any $X\subseteq\Omega _A$, $Y\subseteq\Omega _B$, the joint probability becomes
\begin{equation*}
\Phi _\rho (A_X,B_Y)=\rmtr\!\sqbrac{\rho\sum\brac{C_z\colon z\in f^{-1}(X)\cap g^{-1}(Y)}}=\rmtr\!\sqbrac{\rho C_{f^{-1}(X)\cap g^{-1}(Y)}}
\end{equation*}
As another example, if $A,B\in\iscript$, then $A,B\to\iscripthat$ so $A=f(\iscripthat\,)$, $B=g(\iscripthat\,)$ for surjections $f,g$. We then obtain
\begin{equation*}
\Phi _\rho (A_X,B_Y)=\rmtr\sqbrac{\rho\iscripthat _{f^{-1}(X)\cap g^{-1}(Y)}}=\rmtr\sqbrac{\iscript _{f^{-1}(X)\cap g^{-1}(Y)}(\rho )}
\end{equation*}
We can continue this for many coexistent entities. Moreover, the entities do not need to be of the same type. For instance, suppose
$A,\iscript\to\jscript$ where $A=f(\jscripthat\,)$ and $\iscript =g(\jscript )$. Then we have that
\begin{align*}
\Phi _\rho (A_X,\iscript _Y)&=\Phi _\rho\sqbrac{f(\jscripthat\,)_X,g(\jscript )_Y}=\Phi _\rho\sqbrac{\jscripthat _{f^{-1}(X)},\jscript _{g^{-1}(Y)}}\\
   &=\rmtr\sqbrac{\jscript _{f^{-1}(X)\cap g^{-1}(Y)}(\rho )}
\end{align*}

For $A\in\oscript (H)$ we define the \textit{probability distribution} $\Phi _\rho ^A(X)=\rmtr (\rho A_X)$ for all $X\subseteq\Omega _A$,
$\rho\in\sscript (H)$. In a similar way, if $\iscript\in\rmin (H)$ we define $\Phi _\rho ^\iscript (X)=\rmtr\sqbrac{\iscript _X(\rho )}$ and if $\mscript$ is a $MM$, then $\Phi _\rho ^\mscript (X)=\Phi _\rho ^{\mscripthat}(X)$.

\begin{lem}    
\label{lem33}
If $\alpha$ is an entity and $f\colon\Omega _\alpha\to\Omega$ is a surjection, then $\Phi ^{f(\alpha )}=\Phi ^\alpha\circ f^{-1}$.
\end{lem}
\begin{proof}
We give the proof for $A\in\oscript (H)$ and the proof for other entities is similar. For $x\in\Omega _A$, $\rho\in\sscript (H)$ we obtain
\begin{align*}
\Phi _\rho ^{f(A)}&=\rmtr\sqbrac{\rho f(A)_x}=\rmtr\sqbrac{\rho A_{f^{-1}(x)}}
   =\rmtr\sqbrac{\rho\sum\brac{A_y\colon f(y)=x}}\\
   &=\sum\brac{\rmtr (\rho A_y)\colon f(y)=x}=\sum\brac{\Phi _\rho ^A(y)\colon f(y)=x}\\
   &=\Phi _\rho ^A\sqbrac{f^{-1}(x)}=\Phi _\rho ^A\circ f^{-1}(x)
\end{align*}
The result now follows.
\end{proof}

We now consider sequential products of observables.

\begin{thm}    
\label{thm34}
If $A,B\in\oscript (H)$ and $h\colon\Omega _B\to\Omega$ is a surjection, then $A$, $(B\mid A)$ and $A\circ h(B)$ are parts of $A\circ B$.
\end{thm}
\begin{proof}
Defining $f\colon\Omega _A\times\Omega _B\to\Omega _A$ by $f(x,y)=x$ we have that
\begin{align*}
f(A\circ B)_x&=(A\circ B)_{f^{-1}(x)}=\!\sum\brac{(A\circ B)_{(y,z)}\colon f(y,z)=x}=\!\!\sum _{z\in\Omega _B}\!(A\circ B)_{(x,z)}\\
   &=\sum _{z\in\Omega _B}A_x\circ B_z=A_x\circ 1=A_x
\end{align*}
Thus, $A=f(A\circ B)$ so $A\to A\circ B$. Defining $g\colon\Omega _A\times\Omega _B\to\Omega _B$ by $g(x,y)=y$ we obtain
\begin{align*}
g(A\circ B)_y&=(A\circ B)_{g^{-1}(y)}=\!\!\sum\brac{(A\circ B)_{(x,z)}\colon g(x,z)=y}=\!\!\sum _{x\in\Omega _A}(A\circ B)_{(x,y)}\\
   &=\sum _{x\in\Omega _A}A_x\circ B_y=(B\mid A)_y
\end{align*}
Hence, $(B\mid A)=g(A\circ B)$ so $(B\mid A)\to A\circ B$. Defining $u\colon\Omega _A\times\Omega _B\to\Omega _A\times\Omega$ by
$u(x,y)=(x,h(y))$ we have that
\begin{align*}
\sqbrac{u(A\circ B)}_{(x,y)}&=(A\circ B)_{u^{-1}(x,y)}=(A\circ B)_{(x,h^{-1}(y))}=A_x\circ B_{h^{-1}(y)}\\
   &=A_x\circ h(B)_y=\sqbrac{A\circ h(B)}_{(x,y)}
\end{align*}
It follows that $A\circ h(B)=u(A\circ B)$. Hence, $A\circ h(B)\to A\circ B$.
\end{proof}
 Some results analogous to Theorem~\ref{thm34} hold for other entities.

\begin{exam}{1}  
We consider the simplest nontrivial example of a sequential product $A\circ B$ of observables. Let $A=\brac{a_0,a_1}$, $B=\brac{b_0,b_1}$ be binary (diatomic) observables. Then $\Omega _{A\circ B}=\brac{0,1}\times\brac{0,1}$ and
\begin{equation*}
A\circ B=\brac{a_0\circ b_0,a_1\circ b_0,a_0\circ b_1,a_1\circ b_1}
\end{equation*}
Except in trivial cases, $A\circ B$ has precisely the following nine parts to within equivalence:
\begin{align*}
&A\circ B,\brac{a_0\circ b_0,a_1+a_0\circ b_1},\brac{a_1\circ b_0,a_0+a_1\circ b_1},\brac{a_0\circ b_1,a_1+a_0\circ b_0}\\
&\brac{a_1\circ b_1,a_0+a_1\circ b_0},\brac{a_0\circ b_0+a_1\circ b_0,a_0\circ b_1+a_1\circ b_1},\brac{a_0,a_1}\\
&\brac{a_0\circ b_0+a_1\circ b_1,a_1\circ b_0+a_0\circ b_1},\brac{1}
\end{align*}
Notice that the sixth of the parts is $(B\mid A)$ and the seventh is $A$ as required by Theorem~\ref{thm34}. Each of the parts is a function of
$A\circ B$. The parts listed correspond to the following functions $f_i\colon\brac{0,1}\times\brac{0,1}\to\brac{1,2,3,4}$, $i=1,2,\ldots ,9$.
\vskip 2pc

\hskip 2pc
\begin{tabular}{|c|c|c|c|c|c|c|}
\hline
function&$(0,0)$&\;$(0,1)$\;&$(1,0)$&$(1,1)$\\
\hline
$f_1$&1&2&3&4\\
\hline
$f_2$&1&2&2&2\\
\hline
$f_3$&2&2&1&2\\
\hline
$f_4$&2&1&2&2\\
\hline
$f_5$&2&2&2&1\\
\hline
$f_6$&1&2&1&2\\
\hline
$f_7$&1&1&2&2\\
\hline
$f_8$&1&2&2&1\\
\hline
$f_9$&1&1&1&1\\
\hline
\hline\noalign{\medskip}
\multicolumn{5}{c}{\textbf{Table 1: Function Values}}\\
\end{tabular}
\vskip -1pc \hfill\qedsymbol
\end{exam}
\vskip 1pc

\begin{exam}{2}  
Similar to Example~1, for the two binary instruments $\iscript =\brac{\iscript _0,\iscript _1}$, $\jscript =\brac{\jscript _0,\jscript _1}$ we have the instrument $\iscript\circ\jscript$ with $\Omega _{\iscript\circ\jscript}=\brac{0,1}\times\brac{0,1}$ and
\begin{equation*}
\iscript\circ\jscript =\brac{\iscript _0\circ\jscript _0,\iscript _1\circ\jscript _0,\iscript _0\circ\jscript _1,\iscript _1\circ\jscript _1}
\end{equation*}
The nine parts of $\iscript\circ\jscript$ to within equivalence are:
\begin{align*}
&\iscript\circ\jscript ,\brac{\iscript _0\circ\jscript _0,\iscript _0\circ\jscript _1+\iscript _1\circ\cscript _\jscript},
   \brac{\iscript _1\circ\jscript _0,\iscript _1\circ\jscript _1+\iscript _0\circ\cscript _\jscript}\\
   &\brac{\iscript _0\circ\jscript _1,\iscript _0\circ\jscript _0+\iscript _1\circ\cscript _\jscript},
   \brac{\iscript _1\circ\jscript _1,\iscript _1\circ\jscript _0+\iscript _0\circ\cscript _\jscript},
   \brac{\cscript _\iscript\circ\jscript _0,\cscript _\iscript\circ\jscript _1}\\
   &\brac{\iscript _0\circ\cscript _\jscript ,\iscript _1\circ\cscript _\jscript},
   \brac{\iscript _0\circ\jscript _0+\iscript _1\circ\jscript _1,\iscript _1\circ\jscript _0+\iscript _0\circ\jscript _1},\brac{\cscript _{\iscript\circ\jscript}}
\end{align*}
As in Example~1, the sixth part is $(\jscript\mid\iscript )$, however, unlike the observable case, the seventh part is not $\iscript$. In fact, unlike that case, $\iscript$ is not a part of $(\iscript\circ\jscript )$.\hfill\qedsymbol
\end{exam}

If $A\in\oscript (H)$ the corresponding \textit{L\"uders instrument} $\lscript ^A\in\rmin (H)$ is defined by $\Omega _{\lscript ^A}=\Omega _A$ and $\lscript _x^A(\rho )=A_x^{1/2}\rho A_x^{1/2}$ for all $\rho\in\sscript (H)$. It follows that \cite{lud51}
\begin{equation*}
\lscript _X^A(\rho )=\sum _{x\in X}A_x^{1/2}\rho A_x^{1/2}
\end{equation*}
for all $\rho\in\sscript (H)$, $X\subseteq\Omega _A$. It is easy to check that $(\lscript ^A)^\wedge =A$. Hence, for $B\in\oscript (H)$ we have that $B\to\lscript ^A$ if and only if $B\to A$.

\begin{thm}    
\label{thm35}
{\rm{(a)}}\enspace $\lscript ^{A\circ B}=\lscript ^A\circ\lscript ^B$ if and only if $A_xB_y=B_yA_x$ for all $x\in\Omega _A$, $y\in\Omega _B$.
{\rm{(b)}}\enspace $(\lscript ^{A\circ B})^\wedge =(\lscript ^A\circ\lscript ^B)^\wedge =A\circ B$.
{\rm{(c)}}\enspace An observable $C$ satisfies $C\to\lscript ^A\circ\lscript ^B$ if and only if $C\to A\circ B$.
\end{thm}
\begin{proof}
(a)\enspace For all $\rho\in\sscript (H)$, $(x,y)\in\Omega _A\times\Omega _B$ we have that
\begin{equation}                
\label{eq33}
(\lscript ^A\circ\lscript ^B)_{(x,y)}(\rho )=\lscript _y^B\paren{\lscript _x^A(\rho )}
  =\lscript _y^B(A_x^{1/2}\rho A_x^{1/2})=B_y^{1/2}A_x^{1/2}\rho A_x^{1/2}B_y^{1/2}
\end{equation}
On the other hand,
\begin{align}                
\label{eq34}
(\lscript ^{A\circ B})_{(x,y)}(\rho )&=(A\circ B)_{(x,y)}^{1/2}\rho (A\circ B)_{(x,y)}^{1/2}=(A_x\circ B_y)^{1/2}\rho (A_x\circ B_y)^{1/2}\notag\\
  &=(A_x^{1/2}B_yA_x^{1/2})^{1/2}\rho (A_x^{1/2}B_yA_x^{1/2})^{1/2}
\end{align}
If $A_xB_y=B_yA_x$, we obtain
\begin{align*}
(\lscript ^{A_\circ B})_{(x,y)}(\rho )&=(A_xB_y)^{1/2}\rho (A_xB_y)^{1/2}=B_y^{1/2}A_x^{1/2}\rho A_x^{1/2}B_y^{1/2}\\
  &=(\lscript ^A\circ\lscript ^B)_{(x,y)}(\rho )
\end{align*}
so that $\lscript ^{A\circ B}=\lscript ^A\circ\lscript ^B$. Conversely, if $\lscript ^{A\circ B}=\lscript ^A\circ\lscript ^B$, letting $\rho =\frac{1}{n}\,1$ where $n=\dim H$, we obtain from \eqref{eq33} and \eqref{eq34} that
\begin{equation*}
B_y\circ A_x=B_y^{1/2}A_xB_y^{1/2}=A_x^{1/2}B_yA_x^{1/2}=A_x\circ B_y
\end{equation*}
It follows that $B_yA_x=A_xB_y$ for all $x\in\Omega _A$, $y\in\Omega _B$ \cite{gn01}.
(b)\enspace We have already pointed out that $(\lscript ^{A\circ B})^\wedge =A\circ B$. To show that
$(\lscript ^A\circ\lscript ^B)^\wedge =A\circ B$, applying \eqref{eq33} gives
\begin{align*}
\rmtr\sqbrac{\rho (\lscript ^A\circ\lscript ^B)_{(x,y)}^\wedge}&=\rmtr\sqbrac{(\lscript ^A\circ\lscript ^B)_{(x,y)}(\rho)}
  =\rmtr (\rho A_x^{1/2}B_yA_x^{1/2})\\
  &\rmtr (\rho A_x\circ B_y)=\rmtr\sqbrac{\rho (A\circ B)_{(x,y)}}
\end{align*}
Hence, $(\lscript ^A\circ\lscript ^B)^\wedge =A\circ B$.
(c)\enspace This follows from (b) and Theorem~\ref{thm31}(a).
\end{proof}

\begin{exam}{3}  
We have seen from Theorem~\ref{thm35}(b) that $(\lscript ^A\circ\lscript ^B)^\wedge =(\lscript ^A)^\wedge\circ (\lscript ^B)^\wedge$. We now show that $(\iscript\circ\jscript )^\wedge\ne\iscripthat\circ\jscripthat$ in general. Let $\delta ,\gamma\in\sscript (H)$ and $A,B\in\oscript (H)$. The instruments $\iscript _x(\rho )=\rmtr (\rho A_x)\delta$ and $\jscript _y(\rho )=\rmtr (\rho B_y)\gamma$ are called \textit{trivial instruments} with observables $A,B$ and states $\delta ,\gamma$, respectively \cite{hz12}. We have that
\begin{equation*}
\rmtr (\rho\iscripthat _x)=\rmtr\sqbrac{\iscript _x(\rho )}=\rmtr\sqbrac{\rmtr (\rho A_x)\delta}=\rmtr (\rho A_x)
\end{equation*}
Hence, $\iscripthat =A$ and similarly $\jscripthat =B$. For all $\rho\in\sscript (H)$ we obtain
\begin{align}                
\label{eq35}
\rmtr\!\sqbrac{\rho (\iscript\circ\jscript )_{(x,y)}^\wedge}&=\rmtr\!\sqbrac{(\iscript\circ\jscript )_{(x,y)}(\rho )}
  =\rmtr\!\sqbrac{\jscript _y\paren{\iscript _x(\rho )}}=\rmtr\!\sqbrac{\jscript _y\paren{\rmtr (\rho A_x)\delta}}\notag\\
  &=\rmtr(\rho A_x)\rmtr\sqbrac{\jscript _y(\delta )}=\rmtr (\rho A_x)\rmtr\sqbrac{\rmtr (\delta B_y)\gamma}\notag\\
  &=\rmtr (\rho A_x)\rmtr (\delta B_y)
\end{align}

On the other hand,
\begin{equation}                
\label{eq36}
\rmtr\sqbrac{\rho\iscripthat _x\circ\jscripthat _y}=\rmtr (\rho A_x\circ B_y)
\end{equation}
Since the right hand sides of \eqref{eq35} and \eqref{eq36} are different in general, we conclude that
$(\iscript\circ\jscript )^\wedge\ne\iscripthat\circ\jscripthat$.\hfill\qedsymbol
\end{exam}

We saw in Theorem~\ref{thm35}(a) that $\lscript ^{A\circ B}\ne\lscript ^A\circ\lscript ^B$, in general. The following lemma shows they can differ in a striking way.

\begin{lem}    
\label{lem36}
If $A_x=\ket{\phi _x}\bra{\phi _x}$ and $B_y=\ket{\psi _y}\bra{\psi _y}$ are atomic observables on $H$, then for all $\rho\in\sscript (H)$, there exist numbers $\lambda _{xy}(\rho )\in\sqbrac{0,1}$ with $\sum\limits _{x,y}\lambda _{xy}(\rho )=1$ such that
$\lscript _{(x,y)}^{A\circ B}(\rho )=\lambda _{xy}(\rho )A_x$ and $(\lscript ^A\circ\lscript ^B)_{(x,y)}(\rho )=\lambda _{xy}(\rho )B_y$ for all
$(x,y)\in\Omega _A\times\Omega _B$.
\end{lem}
\begin{proof}
For all $\rho\in\sscript (H)$ we have that
\begin{align*}
(\lscript ^A\circ\lscript ^B)_{(x,y)}(\rho )&=\lscript _y^B\paren{\lscript _x^A(\rho )}=B_yA_x\rho A_xB_y\\
  &=\ket{\psi _x}\bra{\psi _y}\,\ket{\phi _x}\bra{\phi _x}\rho\ket{\phi _x}\bra{\phi _x}\,\ket{\psi _y}\bra{\psi _y}\\
  &=\ab{\elbows{\phi _x,\psi _y}}^2\elbows{\phi _x,\rho\phi _x}B_y
\end{align*}
Since
\begin{equation*}
A_xB_yA_x=\ket{\phi _x}\bra{\phi _x}\,\ket{\psi _y}\bra{\psi _y}\,\ket{\phi _x}\bra{\phi _x}=\ab{\elbows{\phi _x,\psi _y}}^2A_x
\end{equation*}
we obtain
\begin{equation*}
(A_xB_yA_x)^{1/2}=\ab{\elbows{\phi _x,\psi _y}}A_x
\end{equation*}
Hence,
\begin{equation*}
(\lscript ^{A\circ B})_{(x,y)}(\rho )=(A_xB_yA_x)^{1/2}\rho (A_xB_yA_x)^{1/2}=\ab{\elbows{\phi _x,\psi _y}}^2\elbows{\phi _x,\rho\phi _x}A_x
\end{equation*}
Letting $\lambda _{xy}(\rho )=\ab{\elbows{\phi _x,\psi _y}}^2\elbows{\phi _x,\rho\phi _x}$, the result follows
\end{proof}

\section{Composite Systems}  
Let $H_1$ and $H_2$ be Hilbert spaces with $\dim H_1=n_1$ and $\dim H_2=n_2$. If $H_1,H_2$ represent quantum systems, we call $H=H_1\otimes H_2$ a \textit{composite} quantum system. For $a\in\escript (H)$, we define the \textit{reduced effects} $a^1\in\escript (H_1)$, $a^2\in\escript (H_2)$ by $a^1=\tfrac{1}{n_2}\,\rmtr _2(a)$, $a^2=\tfrac{1}{n_1}\,\rmtr _1(a)$. We view $a^i$ to be the effect $a$ as measured in system $i=1,2$. The map $a\mapsto a^1$ is a surjective effect algebra morphism from $\escript (H)$ onto $\escript (H_1)$ and similarly for $a\mapsto a^2$ \cite{gg02,gn01}. Conversely, if $a\in\escript (H_1)$, $b\in\escript (H_2)$, then $a\otimes b\in\escript (H)$ and
\begin{equation*}
(a\otimes b)^1=\tfrac{1}{n_2}\,\rmtr _2(a\otimes b)=\tfrac{1}{n_2}\,\rmtr (b)a
\end{equation*}
Similarly, $(a\otimes b)^2=\tfrac{1}{n_1}\,\rmtr (a)b$. It follows that
\begin{align*}
(a^1\otimes a^2)^1&=\tfrac{1}{n_2}\,\rmtr (a^2)a^1\\
\intertext{and}
(a^1\otimes a^2)^2&=\tfrac{1}{n_1}\,\rmtr (a^1)a^2\\
\end{align*}
An effect $a\in\escript (H)$ is \textit{factorized} if $a=b\otimes c$ for $b\in\escript (H_1)$, $c\in\escript (H_2)$ \cite{hz12}.

\begin{lem}    
\label{lem41}
If $a\in\escript (H)$ with $a\ne 0$, then $a$ is factorized if and only if
\begin{equation}                
\label{eq41}
a=\frac{n_1n_2}{\rmtr (a)}\,a^1\otimes a^2
\end{equation}
\end{lem}
\begin{proof}
If \eqref{eq41} holds, then $a$ is factorized. Conversely, suppose $a$ is factorized with $a=b\otimes c$, $b\in\escript (H_1)$,
$c\in\escript (H_2)$. Then $a^1=\tfrac{1}{n_2}\,\rmtr (c)b$ and $a^2=\tfrac{1}{n_1}\,\rmtr (b)c$. Hence, $b=\tfrac{n _2}{\rmtr (c)}\,a^1$ and
$c=\tfrac{n_1}{\rmtr (b)}\,a^2$. We conclude that
\begin{equation*}
a=\frac{n_1n_2}{\rmtr (b)\rmtr (c)}\,a^1\otimes a^2=\frac{n_1n_2}{\rmtr (a)}\,a^1\otimes a^2\qedhere
\end{equation*}
\end{proof}

\begin{cor}    
\label{cor42}
If $a\in\escript (H)$, then $a=a^1\otimes a^2$ if and only if $a=0$ or $a=1$.
\end{cor}
\begin{proof}
If $a=0$ or $a=1$, then clearly $a=a^1\otimes a^2$. Conversely, if $a=a^1\otimes a^2$, then by Lemma~\ref{lem41}, $a=0$ or
$\rmtr (a)=n_1n_2$. In the latter case, $a=1$.
\end{proof}

An effect is \textit{indecomposable} if it has the form $a=\lambda b$ where $0\le\lambda\le 1$ and $b$ is an atom.

\begin{thm}    
\label{thm43}
Let $a\in\escript (H)$ be an atom $a=P_\psi$ where $H=H_1\otimes H_2$.
{\rm{(a)}}\enspace $a$ is factorized if and only if $a^1$ and $a^2$ are indecomposable.
{\rm{(b)}}\enspace We  can arrange the nonzero eigenvalues $\alpha _1,\alpha _2,\ldots ,\alpha _n$ of $a^1$ and the nonzero eigenvalues
$\beta _1,\beta _2,\ldots ,\beta _n$ of $a^2$ so that $\alpha _i=\tfrac{n_1}{n_2}\,\beta _i$, $i=1,2,\ldots ,n$. Hence, if $n_1=n_2$, then the eigenvalues of $a^1$ and $a^2$ are identical.
\end{thm}
\begin{proof}
The unit vector $\psi\in H$ has a Schmidt decomposition $\psi =\sum\limits _{i=1}^m\lambda _i\psi _i\otimes\phi _i$, $\lambda _i\ge 0$,
$\sum\lambda _i^2=1$. We have that
\begin{align*}
a=\ket{\psi}\bra{\psi}&=\ket{\sum\lambda _i\psi _i\otimes\phi _i}\bra{\sum\lambda _j\psi _j\otimes\phi _j}
  =\sum _{i,j}\lambda _i\lambda _j\ket{\psi _i\otimes\phi _i}\bra{\psi _j\otimes\phi _j}\\
  &=\sum _{i,j}\lambda _i\lambda _j\ket{\psi _i}\bra{\psi _j}\otimes\ket{\phi _i}\bra{\phi _j}
\end{align*}
Hence, 
\begin{align}                
\label{eq42}
a^1&=\tfrac{1}{n_2}\,\rmtr _2(a)
  =\tfrac{1}{n_2}\,\sum _{i,j}\lambda _i\lambda _j\rmtr _2\paren{\ket{\psi _i}\bra{\psi _j}\otimes\ket{\phi _i}\bra{\phi _j}}\notag\\
  &=\tfrac{1}{n_2}\,\sum _{i,j}\lambda _i\lambda _j\delta _{ij}\ket{\psi _i}\bra{\psi _j}=\tfrac{1}{n_2}\,\sum\lambda _i^2P_{\psi _i}
\end{align}
and similarly
\begin{equation}                
\label{eq43}
a^2=\tfrac{1}{n_1}\,\sum\lambda _i^2P_{\phi _i}
\end{equation}
Now $a$ is factorized if and only if $\psi$ is factorized which is equivalent to $m=1$ and $\psi =\psi _1\otimes\phi _1$. Applying \eqref{eq42} and \eqref{eq43} we conclude that $a$ is factorized if and only if $a^1=\tfrac{1}{n_2}\,\lambda _1^2P_{\psi _1}$ and
$a^2=\tfrac{1}{n_1}\,\lambda _1^2P_{\phi _1}$ in which case $a^1$ and $a^2$ are indecomposable. This completes the proof of (a). To prove (b), we see from \eqref{eq42}, \eqref{eq43} that the eigenvalues of $a^1,a^2$ are $\alpha _i=\tfrac{1}{n_2}\,\lambda _i^2$ and
$\beta _i=\tfrac{1}{n_1}\,\lambda _i^2$. It follows that $\alpha _i=\tfrac{n_1}{n_2}\,\beta _i$.
\end{proof}

If $A\in\oscript (H_1\otimes H_2)$ we define the \textit{reduced observables} $A^1\in\oscript (H_1)$, $A^2\in\oscript (H_2)$ by
$A^1=\brac{A_x^1\colon x\in\Omega _A}$ and $A^2=\brac{A_x^2\colon x\in\Omega _A}$. Note that $A^1(A^2)$ is indeed an observable because
\begin{equation*}
\sum _{x\in\Omega _A}A_x^1=\sum _{x\in\Omega _A}\tfrac{1}{n_2}\,\rmtr _2(A_x)=\tfrac{1}{n_2}\,\rmtr _2\paren{\sum _{x\in\Omega _A}A_x}
  =\tfrac{1}{n_2}\,\rmtr _2(1_1\otimes 1_2)=1_1
\end{equation*}

\begin{lem}    
\label{lem44}
If $A\in\oscript (H_1\otimes H_2)$ and $\rho _1\in\sscript (H_1)$, then
\begin{equation*}
\Phi _{\rho _1}^{A^1}=\Phi _{\rho _1\otimes 1_2/n_2}^A
\end{equation*}
\end{lem}
\begin{proof}
For $X\subseteq\Omega _A$ we have that
\begin{align*}
\Phi _{\rho _1}^{A^1}(X)&=\rmtr (\rho _1A_x^1)=\rmtr\sqbrac{\rho _1\tfrac{1}{n_2}\,\rmtr _2(A_X)}
  =\tfrac{1}{n_2}\,\rmtr\sqbrac{\rho _1\rmtr _2(A_X)}\\
  &=\tfrac{1}{n_2}\,\rmtr\sqbrac{A_X(\rho _1\otimes 1_2)}=\rmtr\sqbrac{\paren{\rho _1\otimes\tfrac{1}{n_2}\,1_2}A_X}
    =\Phi _{\rho _1\otimes 1_2/n_2}(X)
\end{align*}
The result now follows.
\end{proof}

In a similar way
\begin{equation*}
\Phi _{\rho _2}^{A^2}=\Phi _{1_1/n_1\otimes\rho _2}^A
\end{equation*}

For $A\in\oscript (H_1)$ we define the $A$-\textit{random measure} on $\Omega _A$ by
\begin{equation*}
\mu ^A(X)=\tfrac{1}{n_1}\,\rmtr (A_X)=\rmtr\paren{\tfrac{1_1}{n_1}\,A_X}=\Phi _{1_1/n_1}^A(X)
\end{equation*}
for all $X\subseteq\Omega _A$. Thus, $\mu ^A$ is the distribution of $A$ in the random state $1_1/n_1$. If $A_1\in\oscript (H_1)$,
$A_2\in\oscript (H_2)$, we define the \textit{composite observable}
\begin{equation*}
B_{(x,y)}=A_{1,x}\otimes A_{2,y}\in\oscript (H_1\otimes H_2)
\end{equation*}
In this case, $\Omega _B=\Omega _{A_1}\times\Omega _{A_2}$ and for $Z\subseteq\Omega _B$ we have that
\begin{equation*}
B_Z=\sum _{(x,y)\in Z}B_{(x,y)}
\end{equation*}
Hence, $B_{X\times Y}=A_{1,X}\otimes A_{2,Y}$.

\begin{lem}    
\label{lem45}
$B_{X\times Y}^1=\mu ^{A^2}(Y)A_{1,X}$ and $B_{X\times Y}^2=\mu ^{A_1}(X)A_{2,Y}$.
\end{lem}
\begin{proof}
For $x\in\Omega _{A_1}$, $y\in\Omega _{A_2}$ we obtain
\begin{equation*}
B_{(x,y)}^1=\tfrac{1}{n_2}\,\rmtr _2\sqbrac{B_{(x,y)}}=\tfrac{1}{n_2}\,\rmtr (A_{1,x}\otimes A_{2,y})=\tfrac{1}{n_2}\,\rmtr (A_{2,y})A_{1,x}
\end{equation*}
Hence,
\begin{equation*}
B_{X\times Y}^1=\tfrac{1}{n_2}\,\rmtr (A_{2,Y})A_{1,X}=\mu ^{A_2}(Y)A_{1,X}
\end{equation*}
The second equation is similar.
\end{proof}

A \textit{transition probability} from $\Omega _1$ to $\Omega _2$ is a map $\nu\colon\Omega _1\times\Omega _2\to\sqbrac{0,1}$ satisfying
$\sum\limits _{y\in\Omega _2}\nu _{xy}=1$ for all $x\in\Omega _1$. (The matrix $\sqbrac{\nu _{xy}}$ is called a \textit{stochastic matrix}.) Let
$A\in\oscript (H_!)$ with outcome-space $\Omega _1$ and let $\nu$ be a transition probability from $\Omega _1$ to $\Omega _2$. Then
$(\nu\tbullet A)_y=\sum\limits _{x\in\Omega _1}\nu _{xy}A_x$ is an observable on $H_1$ with outcome-space $\Omega _2$ called a
\textit{post-processing} of $A$ from $\Omega _1$ to $\Omega _2$ \cite{hrsz09}. If we also have $B\in\oscript (H_2)$ with outcome-space
$\Omega _3$ and $\mu$ a transition probability from $\Omega _3$ to $\Omega _4$, we can form the post-processing $\mu\tbullet B$.

\begin{thm}    
\label{thm46}
{\rm{(a)}}\enspace $(\nu\tbullet A)\otimes (\mu\tbullet B)\in\oscript (H_1\otimes H_2)$ with outcome-space $\Omega _2\times\Omega _4$ and is a post-processing $\alpha\tbullet (A\otimes B)$ from $\Omega _1\times\Omega _3$ to $\Omega _2\times\Omega _4$ where
$\alpha\paren{(x,r),(y,s)}=\nu _{xy}\mu _{rs}$.
{\rm{(b)}}\enspace If $A\in\oscript (H_1\otimes H_2)$, then $(\nu\tbullet A)^1=\nu\tbullet A^1$ and $(\nu\tbullet A)^2=\nu\tbullet A^2$. 
\end{thm}
\begin{proof}
(a)\enspace The map $\alpha\colon\Omega _1\times\Omega _3\to\Omega _2\times\Omega _4$ is a transition probability because
$\alpha _{((x,r),(y,s))}\ge 0$ and
\begin{equation*}
\sum _{(y,s)\in\Omega _2\times\Omega _4}\alpha _{((x,r),(y,s))}=\sum _{y,s}\nu _{xy}\mu _{rs}=1
\end{equation*}
Moreover, $\nu\tbullet A)\otimes (\mu\tbullet B)\in\oscript (H_1\otimes H_2)$ with outcome-space $\Omega _2\times\Omega _4$ and we have that
\begin{align*}
\sqbrac{(\nu\tbullet A)\otimes (\mu\tbullet B)}_{(y,s)}&=(\nu\tbullet A)_y\otimes (\mu\tbullet B)_s\\
  &=\paren{\sum _{x\in\Omega _1}\nu _{xy}A_x}\otimes\paren{\sum _{r\in\Omega _3}\mu _{rs}B_r}\\
  &=\sum _{x\in\Omega _1}\sum _{r\in\Omega _3}\nu _{xy}\mu _{rs}A_x\otimes B_r\\
  &=\sum _{x,r}\alpha _{((x,r),(y,s))}A_x\otimes B_r=\sqbrac{\alpha\tbullet (A\otimes B)}_{(y,s)}
\end{align*}
Hence, $(\nu\tbullet A)\otimes (\mu\tbullet B)=\alpha\tbullet (A\otimes B)$.
(b)\enspace This follows from
\begin{align*}
(\nu\tbullet A)_y^1&=\paren{\sum _x\nu _{xy}A_x}^1=\tfrac{1}{n_2}\,\rmtr _2\paren{\sum _x\nu _{xy}A_x}
  =\tfrac{1}{n_2}\,\sum _x\nu _{xy}\rmtr _2(A_x)\\
  &=\sum _x\nu _{xy}A_x^1=(\nu\tbullet A^1)_y
\end{align*}
That $(\nu\tbullet A)^2=\nu\tbullet A^2$ is similar.
\end{proof}

We have seen in Theorem~\ref{thm32} that coexistence is equivalent to joint measurability. This is used in the next theorem \cite{hmr14}.

\begin{thm}    
\label{thm47}
{\rm{(a)}}\enspace If $A_1,B_1\in\oscript (H_1)$ coexist with joint observable $C_2$, then $A_1\otimes A_2$, $B_1\otimes B_2$ coexist with joint observable $C=C_1\otimes C_2$. 
{\rm{(b)}}\enspace If $A,B\in\oscript (H_1\otimes H_2)$ coexist with joint observable $C$, then $A^1,B^1$ coexist with joint observable $C^1$ and $A^2,B^2$ coexist with joint observable $C^2$.
\end{thm}
\begin{proof}
(a)\enspace We write $C_{1,(x,y)}$ for $(x,y)\in\Omega _{A_1}\times\Omega _{B_1}$ and $C_{2,(x',y')}$ for
$(x',y')\in\Omega _{A_2}\times\Omega _{B_2}$. Then
\begin{equation*}
C_{(x,y,x',y')}=C_{1,(x,y)}\otimes C_{2,(x',y')}
\end{equation*}
and we have that
\begin{equation*}
\sum _{(y,y')}C_{(x,y,x',y')}=\sum _yC_{1,(x,y)}\otimes\sum _{y'}C_{2,(x',y')}=A_{1,x}\otimes A_{2,x'}=(A_1\otimes A_2)_{(x,x')}
\end{equation*}
Moreover,
\begin{equation*}
\sum _{(x,x')}C_{(x,y,x',y')}=\sum _xC_{1,(x,y)}\otimes\sum _{x'}C_{2,(x',y')}=B_{1,y}\otimes B_{2,y'}=(B_1\otimes B_2)_{(y,y')}
\end{equation*}
and the result follows.
(b)\enspace For all $(x,y)\in\Omega _A\times\Omega _B$ we obtain
\begin{align*}
A_x^1=\sqbrac{\sum _yC_{(x,y)}}^1&=\tfrac{1}{n_2}\,\rmtr _2\sqbrac{\sum _yC_{(x,y)}}=\sum _y\sqbrac{\tfrac{1}{n_2}\,\rmtr (C_{(x,y)})}\\
   &=\sum _yC_{(x,y)}^1
\end{align*}
Similarly, $B_y^1=\sum _xC_{(x,y)}^1$ so $A^1,B^1$ coexist with joint observable $C^1$. The result for $A^2,B^2$ is similar.
\end{proof}

For an instrument $\iscript\in\rmin (H_1\otimes H_2)$ on the composite system, the \textit{reduced instrument} on system~1 is defined by
\cite{gud220,gud320}
\begin{equation*}
\iscript _x^1(\rho _1)=\tfrac{1}{n_2}\,\rmtr _2\sqbrac{\iscript _x(\rho _1\otimes 1_2)}
\end{equation*}
for all $\rho _1\in\sscript (H_1)$, $x\in\Omega _\iscript$. Similarly,
\begin{equation*}
\iscript _x^2(\rho _1)=\tfrac{1}{n_1}\,\rmtr _1\sqbrac{\iscript _x(1_1\otimes\rho _2)}
\end{equation*}
for all $\rho _2\in\sscript (H_2)$, $x\in\Omega _\iscript$.

\begin{thm}    
\label{thm48}
$(\iscript ^1)^\wedge =(\iscripthat\,)^1$ and $(\iscript ^2)^\wedge =(\iscripthat\,)^2$.
\end{thm}
\begin{proof}
For all $\rho _1\in\sscript (H_1)$ we have that
\begin{align*}
\rmtr\sqbrac{\rho _1(\iscripthat\,)_x^1}&=\tfrac{1}{n_2}\,\rmtr\sqbrac{\rho _1\rmtr _2(\iscripthat\,)_x}
  =\tfrac{1}{n_2}\,\rmtr\sqbrac{(\rho _1\otimes 1_2)\iscripthat _x}\\
  &=\tfrac{1}{n_2}\,\rmtr\sqbrac{\iscript _x(\rho _1\otimes 1_2)}=\rmtr\sqbrac{\iscript _x^1(\rho _1)}
  =\rmtr\sqbrac{\rho _1(\iscript ^1)_x^\wedge}
\end{align*}
We conclude that $(\iscript ^1)^\wedge =(\iscripthat\,)^1$ and similarly, $(\iscript ^2)^\wedge =(\iscripthat\,)^2$.
\end{proof}

For $\iscript\in\rmin (H_1)$ we define the $\iscript$-\textit{random measure} on $\Omega _\iscript$ by
\begin{equation*}
\mu ^\iscript (X)=\tfrac{1}{n_1}\,\rmtr\sqbrac{\iscript _X(1_1)}
\end{equation*}
For $\iscript _1\in\rmin (H_1)$, $\iscript _2\in\rmin (H_2)$ we define $\jscript =\iscript _1\otimes\iscript _2\in\rmin (H_1\otimes H_2)$ with outcome-space $\Omega _{\iscript _1}\times\Omega _{\iscript _2}$ by $\jscript _{(x,y)}=\iscript _{1,x}\otimes\iscript _{2,y}$. It is easy to check that $\jscript$ is indeed an instrument.

\begin{thm}    
\label{thm49}
Let $\jscript =\iscript _1\otimes\iscript _2\in\rmin (H_1\otimes H_2)$.
{\rm{(a)}}\enspace $\jscript _{(x,y)}^1(\rho _1)=\mu ^{\iscript _2}(y)\iscript _{1,x}(\rho _1)$ for all $\rho _1\in\sscript (H_1)$ and
$\jscript _{(x,y)}^2(\rho _2)=\mu ^{\iscript _1}(x)\iscript _{2,y}(\rho _2)$ for all $\rho _2\in\sscript (H_2)$.
{\rm{(b)}}\enspace $(\iscript _1\otimes\iscript _2)^\wedge =\iscripthat _1\otimes\iscripthat _2$.
\end{thm}
\begin{proof}
(a)\enspace For all $\rho _1\in\sscript (H_1)$ we have that
\begin{align*}
\jscript _{(x,y)}^1(\rho _1)&=\tfrac{1}{n_2}\,\rmtr _2\sqbrac{\jscript _{(x,y)}(\rho _1\otimes 1_2)}
   =\tfrac{1}{n_2}\,\rmtr _2\sqbrac{\iscript _{1,x}\otimes\iscript _{2,y}(\rho\otimes 1_2)}\\
   &=\tfrac{1}{n_2}\,\rmtr _2\sqbrac{\iscript _{1,x}(\rho _1)\otimes\iscript _{2,y}(1_2)}
   =\tfrac{1}{n_2}\,\rmtr\sqbrac{\iscript _{2,y}(1_2)}\iscript _{1,x}(\rho _1)\\
   &=\mu ^{\iscript _2}(y)\iscript _{1,x}(\rho _1)
\end{align*}
Similarly, $\jscript _{(x,y)}^2(\rho _2)=\mu ^{\iscript _1}(x)\iscript _{2,y}(\rho _2)$ for all $\rho _2\in\sscript (H_2)$.
(b)\enspace For all $\rho _1\in\sscript (H_1)$, $\rho _2\in\sscript (H_2)$ we have that
\begin{align*}
\rmtr\!\sqbrac{\rho _1\otimes\rho _2(\iscript _1\otimes \iscript _2)_{(x,y)}^\wedge}
  &=\rmtr\!\sqbrac{\iscript _{1,x}\otimes\iscript _{2,y}(\rho _1\otimes\rho _2)}
  =\rmtr\!\sqbrac{\iscript _{1,x}(\rho _1)\otimes\iscript _{2,y}(\rho _2)}\\
  &=\rmtr\!\sqbrac{\iscript _{1,x}(\rho _1)}\rmtr\sqbrac{\iscript _{2,y}(\rho _2)}
  =\rmtr\sqbrac{\rho _1\iscripthat _{1,x}}\rmtr\sqbrac{\rho _2\iscripthat _{2,y}}\\
  &=\rmtr\sqbrac{\rho _1\otimes\rho _2(\iscripthat _{1,x}\otimes\iscripthat _{2,y})}
\end{align*}
and the result follows.
\end{proof}

A \textit{Kraus instrument} is an instrument of the form $\iscript _x(\rho )=S_x\rho S_x^*$ where $\sum\limits _xS_x^*S_x=1$,
$x\in\Omega _\iscript$. The operators $S_x$ are called \textit{Kraus operators} for $\iscript$ \cite{kra83}.

\begin{lem}    
\label{lem410}
Let $\iscript _1\in\rmin (H_1)$, $\iscript _2\in\rmin (H_2)$ be Kraus instruments with operators $S_{1,x}$, $S_{2,y}$, respectively.
{\rm{(a)}}\enspace $\jscript =\iscript _1\otimes\iscript _2$ is a Kraus instrument with operators $S_{1,x}\otimes S_{2,y}$.
{\rm{(b)}}\enspace $\jscript ^1,\jscript ^2$ are Kraus instruments with operators
\begin{align*}
T_{(x,y)}&=\sqbrac{\tfrac{1}{n_2}\,\rmtr (S_{2,y}S_{2,y}^*)}^{1/2}S_{1,x}\\
R_{(x,y)}&=\sqbrac{\tfrac{1}{n_1}\,\rmtr (S_{1,x}S_{1,x}^*)}^{1/2}S_{2,y}
\end{align*}
\end{lem}
\begin{proof}
(a)\enspace For all $\rho _1\in\sscript (H_1)$, $\rho _2\in\sscript (H_2)$ we have that
\begin{align*}
\jscript _{(x,y)}(\rho _1\otimes\rho _2)&=(\iscript _{1,x}\times\iscript _{2,y})(\rho _1\otimes\rho _2)
  =\iscript _{1,x}(\rho _1)\otimes\iscript _{2,y}(\rho _2)\\
  &=S_{1,x}\rho _1S_{1,x}^*\otimes S_{2,y}\rho _2S_{2,y}^*\\
  &=S_{1,x}\otimes S_{2,y}(\rho _1\otimes\rho _2)S_{1,x}^*\otimes S_{2,y}^*
\end{align*}
and the result follows. (b)\enspace For $\rho\in\sscript (H_1)$ we obtain
\begin{equation*}
\jscript _{(x,y)}^1(\rho _1)=\tfrac{1}{n_2}\,\rmtr\sqbrac{\iscript _{2,y}(1_2)}\iscript _{1,x}(\rho _1)
  =\tfrac{1}{n_2}\,\rmtr (S_{2,y}S_{2,y}^*)S_{1,x}\rho _1S_{1,x}^*
\end{equation*}
This can be considered to be a Kraus instrument with operators $T_{(x,y)}$ given above. The result for $\jscript ^2$ is similar.
\end{proof}

Notice that a L\"uders instrument defined by $\lscript _x^A(\rho _1)=A_x^{1/2}\rho _1A_x^{1/2}$ for all $\rho _1\in\sscript (H_1)$ is a particular case of a Kraus instrument with operators $A_x^{1/2}$ \cite{lud51}.

\begin{cor}    
\label{cor41}
Let $A\in\oscript (H_1)$, $B\in\oscript (H_2)$.
{\rm{(a)}}\enspace $\lscript _x^A\otimes\lscript _y^B=\lscript _{(x,y)}^{A\otimes B}$.
{\rm{(b)}}\enspace $(\lscript _x^A\otimes\lscript _y^B)^1=\lscript _{(x,y)}^C$ where $C=\tfrac{1}{n_2}\,\rmtr (B_y)A_x$ and
$(\lscript _x^A\otimes\lscript _y^B)^2=\lscript _{(x,y)}^D$ where $D=\tfrac{1}{n_2}\,\rmtr (A_x)B_y$.
\end{cor}

We say that a Kraus instrument $\iscript\in\rmin (H_1\otimes H_2)$ with operators $R_x$ is \textit{factorized} if $R_x=S_x\otimes T_x$ for all
$x\in\Omega _\iscript$. We conjecture that if $\iscript\in\rmin (H_1\otimes H_2)$ is Kraus, then $\iscript ^1$ and $\iscript ^2$ need not be Kraus. However, we do have the following result.

\begin{lem}    
\label{lem412}
If $\iscript\in\rmin (H_1\otimes H_2)$ is Kraus and factorized, then $\iscript ^1$ and $\iscript ^2$ are Kraus.
\end{lem}
\begin{proof}
If the operators $R_x$ for $\iscript$ satisfy $R_x=S_x\otimes T_x$, then for all $\rho _1\in\sscript (H_1)$ we have that
\begin{align*}
\iscript _x^1(\rho _1)&=\tfrac{1}{n_2}\,\rmtr _2\sqbrac{\iscript _x(\rho _1\otimes 1_2)}
  =\tfrac{1}{n_2}\,\rmtr _2\sqbrac{R_x(\rho _1\otimes 1_2)R_x^*}\\
  &=\tfrac{1}{n_2}\,\rmtr _2\sqbrac{S_x\otimes T_x(\rho _1\otimes 1_2)S_x^*\otimes T_x^*}\\
  &=\tfrac{1}{n_2}\,\rmtr _2(S_x\rho _1S_x^*\otimes T_xT_x^*)=\tfrac{1}{n_2}\,\rmtr (T_xT_x^*)S_x\rho _1S_x^*
\end{align*}
Hence, $\iscript ^1$ is Kraus with operators $\sqbrac{\tfrac{1}{n_2}\,\rmtr (T_xT_x^*)}^{1/2}S_x$. Similarly, $\iscript ^2$ is Kraus with operators
$\sqbrac{\tfrac{1}{n_1}\,\rmtr (S_xS_x^*)}^{1/2}T_x$.
\end{proof}

We do not know if the converse of Lemma~\ref{lem412} holds. We now consider trivial instruments (see Example~3).

\begin{lem}    
\label{lem413}
Let $\iscript _1\in\rmin (H_1)$, $\iscript _2\in\rmin (H_2)$ be trivial instruments with
\begin{equation*}
\iscript _{1,x}(\rho _1)=\rmtr (\rho _1A_x)\alpha ,\quad\iscript _{2,y}(\rho _2)=\rmtr (\rho _2B_y)\beta
\end{equation*}
{\rm{(a)}}\enspace $\iscript _1\otimes\iscript _2\in\rmin (H_1\otimes H_2)$ is trivial with observable $A\otimes B$ and state
$\alpha\otimes\beta$
{\rm{(b)}}\enspace $(\iscript _1\otimes\iscript _2)^1$, $(\iscript _1\otimes\iscript _2)^2$ are trivial with observables $\mu ^B(y)A_x$,
$\mu ^A(x)B_y$ and states $\alpha ,\beta$, respectively.
\end{lem}
\begin{proof}
(a)\enspace For all $(x,y)\in\Omega _{\iscript _1}\times\Omega _{\iscript _2}$, $\rho _1\in\sscript (H_1)$, $\rho _2\in\sscript (H_2)$ we have that
\begin{align*}
(\iscript _1\otimes\iscript _2)_{(x,y)}(\rho _1\otimes\rho _2)&=\iscript _{1,x}(\rho _1)\otimes\iscript _{2,y}(\rho _2)
   =\rmtr (\rho _1A_x)\alpha\otimes\rmtr (\rho _2B_y)\beta\\
   &=\rmtr (\rho _1A_x)\rmtr (\rho _2B_y)\alpha\otimes\beta =\rmtr (\rho _1A_x\otimes\rho _2B_y)\alpha\otimes\beta\\
   &=\rmtr (\rho _1\otimes\rho _2A\otimes B_{(x,y)})\alpha\otimes\beta
\end{align*}
The result now follows. (b)\enspace This follows from
\begin{align*}
(\iscript _1\otimes\iscript _2)_{(x,y)}^1(\rho _1)&=\tfrac{1}{n_2}\,\rmtr _2\sqbrac{\iscript _{1,x}\otimes\iscript _{2,y}(\rho _1\otimes 1_2)}
  =\tfrac{1}{n_2}\,\rmtr _2\sqbrac{\iscript _{1,x}(\rho _1)\otimes\iscript _{2,y}(1_2)}\\
  &=\tfrac{1}{n_2}\,\rmtr\sqbrac{\iscript _{2,y}(1_2)}\iscript _{1,x}(\rho _1)
  =\tfrac{1}{n_2}\,\rmtr\sqbrac{\rmtr (1_2B_y)\beta}\rmtr (\rho _1A_x)\alpha\\
  &=\tfrac{1}{n_2}\,\rmtr (B_y)\rmtr (\rho _1A_x)\alpha =\rmtr\sqbrac{\rho _1\mu ^B(y)A_x}\alpha
\end{align*}
and similarly
\begin{equation*}
(\iscript _1\otimes\iscript _2)_{(x,y)}^2(\rho _2)=\rmtr\sqbrac{\rho _2\mu ^A(x)B_y}\beta\qedhere
\end{equation*}
\end{proof}

\begin{lem}    
\label{lem414}
Let $\iscript\in\rmin (H_1\otimes H_2)$ be trivial with $\iscript _x(\rho )=\rmtr (\rho A_x)\alpha$.
{\rm{(a)}}\enspace $\iscript ^1$, $\iscript ^2$ are trivial with observables $A_x^1$, $A_x^2$ and states $\rmtr _2(\alpha )$, $\rmtr _1(\alpha )$, respectively.
{\rm{(b)}}\enspace Letting $\jscript =\iscript ^1\otimes\iscript ^2$ we have that $\jscript$ is trivial with observable $A^1\otimes A^2$ and state
$\rmtr _2(\alpha )\otimes\rmtr _1(\alpha )$. Moreover, $\jscript _{(x,y)}^1=\iscript _x^1$ and $\jscript _{(x,y)}^2=\iscript _y^2$ for all
$(x,y)\in\Omega _\jscript$.
\end{lem}
\begin{proof}
(a)\enspace For all $\rho _1\in\sscript (H_1)$ and $x\in\Omega _\iscript$ we have that
\begin{align*}
\iscript _x^1(\rho _1)&=\tfrac{1}{n_2}\,\rmtr _2\sqbrac{\iscript _x(\rho _1\otimes 1_2)}
  =\tfrac{1}{n_2}\,\rmtr _2\brac{\rmtr\sqbrac{(\rho _1\otimes 1_2)A_x}\alpha}\\
  &=\tfrac{1}{n_2}\,\rmtr\sqbrac{(\rho _1\otimes 1_2)A_x}\rmtr _1(\alpha )=\tfrac{1}{n_2}\,\rmtr\sqbrac{\rmtr _2(A_x)\rho _1}\rmtr _2(\alpha )\\
  &=\rmtr\sqbrac{\rho _1\,\tfrac{1}{n_2}\,\rmtr _2(A_x)}\rmtr _2(\alpha )=\rmtr (\rho _1A_x^1)\rmtr _2(\alpha )
\end{align*}
Similarly, $\iscript _x^2(\rho _2)=\rmtr (\rho _2A_x^2)\rmtr _1(\alpha )$ so the result follows.
(b)\enspace This result follows from Lemma~\ref{lem413}(b).
\end{proof}

We now consider $MM$s for composite systems. A \textit{single probe} $MM$ on $H=H_1\otimes H_2$ has the form
$\mscript =(H,K,\eta ,\nu ,F)$ as defined before. As discussed earlier, $\mscripthat\in\rmin (H)$ is the instrument measured by $\mscript$.
Then $\mscripthat ^1\in\rmin (H_1)$ and for $\rho _1\in\sscript (H_1)$ we obtain
\begin{align}                
\label{eq44}
\mscripthat\, _x^1 (\rho _1)&=\tfrac{1}{n_2}\,\rmtr _2\sqbrac{\mscripthat _x(\rho _1\otimes 1_2)}\notag\\
   &=\tfrac{1}{n_2}\,\rmtr _2\brac{\rmtr _K\sqbrac{\nu (\rho _1\otimes 1_2\otimes\eta )(1_1\otimes 1_2\otimes F_x)}}
\end{align}
We have a similar expression for $\mscripthat ^2\in\rmin (H_2)$.

Corresponding to $\mscript$ we define the \textit{reduced} $MM$ $\mscript _1=(H_1,K,\eta ,\nu _1,F)$ where
$\nu _1\in\sscript (H_1\otimes K)$ is given by
\begin{equation*}
\nu _1(\rho _1\otimes\eta )=\tfrac{1}{n_2}\,\rmtr _2\sqbrac{\nu (\rho _1\otimes 1_2\otimes\eta )}
\end{equation*}
We then have for $\rho _1\in\sscript (H_1)$ that
\begin{align}                
\label{eq45}
\mscripthat _{1,x}(\rho _1)&=\rmtr _K\sqbrac{\nu _1(\rho _1\otimes\eta )(1_1\otimes F_x)}\notag\\
   &=\tfrac{1}{n_2}\,\rmtr _K\brac{\rmtr _2\sqbrac{\nu (\rho _1\otimes 1_2\otimes\eta )}(1_1\otimes F_x)}
\end{align}
Similarly, we define $\mscript _2=(H_2,K,\eta ,\nu _2,F)$ and an analogous formula for $\mscripthat _2$. Notice that \eqref{eq44} and
\eqref{eq45} are quite similar and they are essentially an interchange of the two partial traces. We now show that they coincide.

\begin{thm}    
\label{thm415}
{\rm{(a)}}\enspace Let $H_1,H_2,H_3$ be finite-dimensional Hilbert spaces and let $A\in\lscript (H_1\otimes H_2\otimes H_3)$, $B\in\lscript (H_3)$. Then
\begin{equation}                
\label{eq46}
\rmtr _2\sqbrac{\rmtr _3\paren{A(1_1\otimes 1_2\otimes B)}}=\rmtr _3\sqbrac{\paren{\rmtr _2(A)}(1_1\otimes B)}
\end{equation}
{\rm{(b)}}\enspace $\mscripthat ^1=\mscripthat _1$ and $\mscripthat ^2=\mscripthat _2$.
\end{thm}
\begin{proof}
(a)\enspace First suppose that $A=A_1\otimes A_2\otimes A_3$ is factorized. We then obtain
\begin{align*}
\rmtr _2\sqbrac{\rmtr _3\paren{A(1_1\otimes 1_2\otimes B)}}
  &=\rmtr _2\sqbrac{\rmtr _3\paren{A_1\otimes A_2\otimes A_3(1_1\otimes 1_2\otimes B)}}\\
  &=\rmtr _2\sqbrac{\rmtr _3(A_1\otimes A_2\otimes A_3B)}\\
  &=\rmtr _2\sqbrac{A_1\otimes A_2\rmtr (A_3B)}=\rmtr (A_3B)\rmtr _2(A_1\otimes A_2)\\
  &=\rmtr (A_3B)\rmtr (A_2)A_1=\rmtr (A_2)\rmtr _3(A_1\otimes A_3B)\\
  &=\rmtr _3\sqbrac{\rmtr (A_2)(A_1\otimes A_3)(1_1\otimes B)}\\
  &=\rmtr _3\sqbrac{\paren{\rmtr _2(A_1\otimes A_2\otimes A_3)}(1_1\otimes B)}\\
  &=\rmtr _3\sqbrac{\paren{\rmtr _2(A)}(1_1\otimes B)}
\end{align*}
Hence, \eqref{eq46} holds when $A$ is factorized. Since any $A\in\lscript (H_1\otimes H_2\otimes H_3)$ is a linear combination of factorized operators, \eqref{eq46} holds in general.
(b)\enspace Letting $A=\nu (\rho _1\otimes 1_2\otimes\eta )$, $B=F_x$ and $K=H_3$ in \eqref{eq46}, we conclude that \eqref{eq44} and \eqref{eq45} coincide. Hence, $\mscripthat ^1=\mscripthat _1$ and similarly, $\mscripthat ^2=\mscripthat _2$.
\end{proof}

We have considered single probe composite $MM$s. We now briefly discuss general composite $MM$s. Let
$\mscript _i=(H_i,K_i,\eta _i,\nu _i,F_i)$, $i=1,2$ be two $MM$s. Define the unitary \textit{swap operator} \cite{hz12}
\begin{align*}
U\colon H_1\otimes H_2\otimes K_1\otimes K_2\to H_1\otimes K_1\otimes H_2\otimes K_2\\
\intertext{by}
U(\phi _1\otimes\phi _2\otimes\psi _1\otimes\psi _2)=\phi _1\otimes\psi _1\otimes\phi _2\otimes\psi _2
\end{align*}
We now define the channel $\nu _1\otimes\nu _2\in\cscript (H_1\otimes H_2\otimes K_1\otimes K_2)$ by
\begin{equation}                
\label{eq47}
\nu _1\otimes\nu _2(\rho _1\otimes\rho _2\otimes\eta _1\otimes\eta _2)
  =U^*\sqbrac{\nu _1(\rho _1\otimes\eta _2)\otimes\nu _2(\rho _2\otimes\eta _2)}U
\end{equation}
The \textit{composite} of $\mscript _1$ and $\mscript _2$ is declared to be
\begin{equation*}
\mscript =\mscript _1\otimes\mscript _2=(H_1\otimes H_2,K_1\otimes K_2,\eta _1\otimes\eta _2,\nu _1\otimes\nu _2,F_1\otimes F_2)
\end{equation*}
For $\rho\in\sscript (H_1\otimes H_2)$ we have that
\begin{equation}                
\label{eq48}
\mscripthat _{(x,y)}(\rho )
  =\rmtr _{K_1\otimes K_2}\sqbrac{\nu _1\otimes\nu _2(\rho\otimes\eta _1\otimes\eta _2)(1_1\otimes 1_2\otimes F_{1,x}\otimes F_{2,y})}
\end{equation}
The next result shows that $\mscript$ has desirable properties.

\begin{thm}    
\label{thm416}
{\rm{(a)}}\enspace For $\rho =\rho _1\otimes\rho _2\in\sscript (H_1\otimes H_2)$ we have
\begin{equation*}
\mscripthat _{(x,y)}(\rho )=\mscripthat _{1,x}(\rho _1)\otimes\mscripthat _{2,y}(\rho _2)
\end{equation*}
{\rm{(b)}}\enspace Defining $\mscripthat ^1$ and $\mscripthat ^2$ in the usual way we obtain
\begin{align*}
\mscripthat _{(x,y)}^1(\rho _1)&=\tfrac{1}{n_2}\,\rmtr\sqbrac{\mscripthat _{2,y}(1_2)}\mscripthat _{1,x}(\rho _1)\\
\intertext{and}
\mscripthat _{(x,y)}^2(\rho _2)&=\tfrac{1}{n_1}\,\rmtr\sqbrac{\mscripthat _{1,x}(1_1)}\mscripthat _{2,y}(\rho _2)
\end{align*}
for all $\rho _1\in\sscript (H_1)$, $\rho _2\in\sscript (H_2)$.
\end{thm}
\begin{proof}
(a)\enspace Applying \eqref{eq47} and \eqref{eq48} we obtain
\begin{align*}
&\mscripthat _{(x,y)}(\rho )\\
  &\quad =\rmtr _{K_1\otimes K_2}\left[U^*\nu _1(\rho _1\otimes\eta _1)\otimes\nu _2(\rho _2\times\eta _2)\right.
  \left. U(1_1\otimes 1_2\otimes F_{1,x}\otimes F_{2,y})\right]\\
  &\quad =\rmtr _{K_1\otimes K_2}\left[\nu _1(\rho _1\otimes\eta _1)\otimes\nu _2(\rho _2\otimes\eta _2)\right.
  \left. U(1_1\otimes 1_2\otimes F_{1,x}\otimes F_{2,y})U^*\right]\\
  &\quad =\rmtr _{K_1}\rmtr _{K_2}\left[\nu _1(\rho _1\otimes\eta _1)\otimes\nu _2(\rho _2\otimes\eta _2)\right.
  \left. (1_1\otimes F_{1,x}\otimes 1_2\otimes F_{2,y})\right]\\
  &\quad =\rmtr _{K_1}\rmtr _{K_2}\left[\nu _1(\rho _1\otimes\eta _1)(1_1\otimes F_{1,x})\right.
  \left.\otimes\nu _2(\rho _2\otimes\eta _2)(1_2\otimes F_{2,y})\right]\\
  &\quad =\mscripthat _{1,x}(\rho _1)\otimes\mscripthat _{2,y}(\rho _2)
\end{align*}
(b)\enspace For $\rho _1\in\sscript (H_1)$ we have that
\begin{align*}
\mscript _{(x,y)}^1(\rho _1)&=\tfrac{1}{n_2}\,\rmtr _2\sqbrac{\mscripthat _{(x,y)}(\rho _1\otimes 1_2)}
  =\tfrac{1}{n_2}\,\rmtr _2\sqbrac{\mscripthat _{1,x}(\rho _1)\otimes\mscripthat _{2,y}(1_2)}\\
  &=\tfrac{1}{n_2}\,\rmtr\sqbrac{\mscripthat _{2,y}(1_2)}\mscripthat _{1,x}(\rho _1)
\end{align*}
The expression for $\mscripthat ^2$ is similar.
\end{proof}

\end{document}